\documentclass[11pt,reqno]{amsart}
\usepackage{amsmath,amssymb}
\usepackage{graphicx, color}
\usepackage{slashed}
\newcommand{\RN}[1]{%
  \textup{\uppercase\expandafter{\romannumeral#1}}%
}

\newcommand{\R}{\mathbb{R}}
\newtheorem{theorem}{Theorem}
\newtheorem{proposition}[theorem]{Proposition}

\newtheorem{lemma}[theorem]{Lemma}

\theoremstyle{definition}

\begin{document} 

\title[The Maxwell Equations on Schwarzschild-de Sitter Spacetimes]{Decay of Solutions to the Maxwell Equations on Schwarzschild-de Sitter Spacetimes}

\author{Jordan Keller}
\address{Jordan Keller\\
Department of Mathematics\\
Columbia University, USA}
\email{keller@math.columbia.edu}

\thanks{The author would like to thank Pei-Ken Hung, Karsten Gimre, Mu-Tao Wang, and Shing-Tung Yau for for their interest in this work.  Especially, he thanks Pei-Ken Hung for many stimulating conversations.} 

\begin{abstract}
In this work, we consider solutions of the Maxwell equations on the Schwarzschild-de Sitter family of black hole spacetimes.  We prove that, in the static region bounded by black hole and cosmological horizons, solutions of the Maxwell equations decay to stationary Coulomb solutions at a super-polynomial rate, with decay measured according to ingoing and outgoing null coordinates.  Our method employs a differential transformation of Maxwell tensor components to obtain higher-order quantities satisfying a Fackerell-Ipser equation, in the style of Chandrasekhar \cite{Chandra1} and the more recent work of Pasqualotto \cite{Pasqualotto}.  The analysis of the Fackerell-Ipser equation is accomplished by means of the vector field method, with decay estimates for the higher-order quantities leading to decay estimates for components of the Maxwell tensor.

\end{abstract}
\maketitle

\section{Introduction}

The Schwarzschild-de Sitter family, parametrized by mass $M > 0$, consists of those spherically symmetric spacetimes solving the Einstein equations with positive cosmological constant $\Lambda$:
\begin{equation}\label{EinsteinCosmo}
Ric(g) = \Lambda g.
\end{equation}
Such spacetimes display a mixture of geometric features: far from the black hole, they resemble the de Sitter spacetime with cosmological constant $\Lambda$; close to the black hole, they take on the characteristics of the Schwarzschild family.

The stability of the Schwarzschild-de Sitter spacetimes as solutions of the Einstein equations \eqref{EinsteinCosmo} was resolved in the recent breakthrough of Hintz and Vasy \cite{HV1}, where the authors prove a more general result on stability of the small angular momenta Kerr-de Sitter spacetimes.  The authors' result is a culmination of a great deal of work on the analysis of hyperbolic equations on Kerr de-Sitter spacetimes within the framework of the Melrose $b$-calculus; see \cite{Vasy, WunschZworski, Dyatlov, HV3, Hintz, HV2}.  In particular, the aforementioned authors prove exponential decay for solutions of the Maxwell equations on Kerr de-Sitter spacetimes with small angular momenta in \cite{HV4}.

There is a comparative dearth of analysis utilizing the vector-field multiplier method, with notable results of Dafermos and Rodnianski on the scalar wave \cite{DRdS} and of Schlue on the cosmological region \cite{Schlue1, Schlue2}.  The present paper adds to this literature, providing boundedness and decay estimates for solutions of the Maxwell equations using red-shift and Morawetz multipliers, along with the static multiplier.  In addition, this work serves as a ``warm-up" exercise, towards a demonstration of the linear stability of the Schwarzschild-de Sitter family by means of vector-field methods.

\section{Schwarzschild-de Sitter Spacetimes}

Regarding the cosmological constant $\Lambda > 0$ as fixed, the Schwarzschild-de Sitter spacetimes comprise a one-parameter family of solutions $(\mathcal{M},g_{M,\Lambda})$ to the Einstein equations
\begin{equation}\label{EinsteinEqn}
Ric(g) = \Lambda g.
\end{equation}
The family is parametrized by mass $M$, which we assume to satisfy the sub-extremal condition
\begin{equation}\label{subextremal}
0 < M < \frac{1}{3\sqrt{\Lambda}}.
\end{equation}

These spacetimes have both black hole and cosmological regions, bounded by respective horizons $\mathcal{H}$ and $\overline{\mathcal{H}}$.  Our primary interest is the region between the two, wherein the spacetimes are static and spherically symmetric.  The staticity and spherical symmetry are encoded by the static Killing field, denoted $T$, and the angular Killing fields, denoted $\Omega_{i}$, with $i = 1,2,3$.  We collect the angular Killing fields in the set $\Omega := \{ \Omega_{i} | i = 1,2,3\}$.

For further details on the Schwarzschild-de Sitter family, we refer the reader to \cite{Carter, Hawking}.

\subsection{Coordinate Systems}

Our results concern the static region, up to and including the future event horizon and the future cosmological horizon.  In the course of our analysis, various coordinate systems will prove useful; we enumerate them below.

In the coordinates $(t,r,\theta,\phi)$, this region has geometry encoded by
\begin{equation}
g_{M,\Lambda} = -(1-\mu)dt^2 + (1-\mu)^{-1}dr^2 + \slashed{g}_{AB}dx^{A}dx^{B},
\end{equation}
with
\begin{align}
\mu &:= \frac{2M}{r}+\frac{1}{3}\Lambda r^2,\\
\slashed{g}_{AB}dx^{A}dx^{B} &:= r^2d\sigma_{S^2} = r^2\left(d\theta^2 + \sin^2\theta d\phi^2\right),
\end{align}
where $d\sigma_{S^2}$ denotes the round metric on the unit sphere.  Note that $\slashed{g}_{AB}dx^{A}dx^{B}$ is the induced metric on the sphere of symmetry $S^2(t,r)$.  As an additional piece of notation, we use $\slashed{\epsilon}_{AB}$ to denote the associated area form on the sphere of symmetry $S^2(t,r)$.  

This static chart is valid for radii $0 < r_{b}< r < r_{c}$, with $r_{b}$ and $r_{c}$ the black hole and cosmological radii appearing as roots of the equation $1 - \mu = 0$.  Note that the equation has a remaining negative root, which we denote by $r_{-}$.  Concretely, we have \cite{LakeRoeder}
\begin{align}
\begin{split}
r_{b} &= \frac{2}{\sqrt{\Lambda}}\cos(\xi/3),\\
r_{c} &= \frac{2}{\sqrt{\Lambda}}\cos(\xi/3 + 4\pi/3),\\
r_{-} &= \frac{2}{\sqrt{\Lambda}}\cos(\xi/3 + 2\pi/3),
\end{split}
\end{align}
where $\xi$ is specified by the relation
\begin{equation}
\cos \xi = -3M\sqrt{\Lambda}.
\end{equation}
In the sub-extremal regime \eqref{subextremal}, the radii $r_{b}$ and $r_{c}$ satisfy
\begin{equation}
0 < 2M < r_{b} < 3M < \frac{1}{\sqrt{\Lambda}} < r_{c} < \frac{3}{\sqrt{\Lambda}} < \infty.
\end{equation}

Letting
\[ \kappa_{b} := \frac{d}{dr}(1-\mu)\Big|_{r = r_{b}},\]
with similar definitions relating to $r_{c}$ and $r_{-}$,  we define the Regge-Wheeler coordinate $r_{*}$ by
\begin{equation}\label{rStar}
r_{*} := -\frac{1}{2\kappa_{c}}\log\left|\frac{r}{r_{c}}-1\right| + \frac{1}{2\kappa_{b}}\log\left|\frac{r}{r_{b}}-1\right| + \frac{1}{2\kappa_{-}}\log\left|\frac{r}{r_{-}}-1\right| + C,
\end{equation}
with $C$ an arbitrary constant.  For convenience in the subsequent analysis, we choose this normalization constant such that $r_{*} = 0$ on the photon sphere $r = 3M$.  In the Regge-Wheeler coordinates, the metric takes on the form
\begin{equation}
g_{M,\Lambda} = -(1-\mu)dt^2 + (1-\mu)dr_{*}^2 + \slashed{g}_{AB}dx^{A}dx^{B}.
\end{equation}

Using the Regge-Wheeler coordinates $(t,r_{*})$, we define the inward and outward null coordinates $(u,v)$ by
\begin{align}\label{uvDef}
\begin{split}
u &= \frac{1}{2}(t - r_{*}),\\
v &= \frac{1}{2}(t + r_{*}),
\end{split}
\end{align}
in which the metric has the form
\begin{equation}
g_{M,\Lambda} = -4(1-\mu)dudv + \slashed{g}_{AB}dx^{A}dx^{B}.
\end{equation}

The pair $(u,v)$, referred to as Eddington-Finkelstein coordinates, break down at either of the horizons.  However, there are well-known, though rather cumbersome, rescalings of $u$ and $v$ which extend regularly to each of the horizons; see \cite{Carter, Hawking}.

For a given pair of null coordinates $(\tilde{u},\tilde{v})$, we define the null hypersurfaces:
\begin{align}\label{nullHypersurfaces}
\begin{split}
\overline{C}_{\tilde{u},\tilde{v}}&:= \{(u,v,\theta,\phi), u = \tilde{u}, v \geq \tilde{v}\},\\
 \underline{C}_{\tilde{u},\tilde{v}} &:= \{ (u,v,\theta,\phi), u \geq \tilde{u}, v = \tilde{v}\},\\
C_{\tilde{u},\tilde{v}} &:= \overline{C}_{\tilde{u},\tilde{v}} \cup \underline{C}_{\tilde{u},\tilde{v}}.
\end{split}
\end{align}

Throughout this work, we use the following index notation: lowercase Latin characters $a,b = 0,1,2,3$ for spacetime indices, and uppercase Latin characters $A,B = 2,3$ for spherical indices.

\subsection{Trapped Null Geodesics}
In this subsection, we recount the well-known phenomenon of null geodesic trapping at the photon sphere.  Such trapping manifests as a ``loss of derivatives" in the integrated decay estimates appearing later in this work, as first described by Ralston \cite{Ralston}.  

Generally, given a Killing field $K^{a}$ and a geodesic $\gamma$ on a pseudo-Riemannian manifold with metric $g_{ab}$, application of the Killing field yields a constant of motion
\[C = g_{ab}K^{a}\dot{\gamma}^{b}\]
along the geodesic.

Specializing to the Schwarzschild-de Sitter setting, we have constants of motion
\begin{align*}
e &= g_{ab}T^{a}\dot{\gamma}^{b},\\
l_{i} &= g_{ab}\Omega_{i}^{a}\dot{\gamma}^{b}.
\end{align*}

Written with respect to the static chart, we have
\begin{align*}
T &= \partial_{t},\\
\Omega_1 &= \partial_{\phi},\\
\Omega_2 &= -\sin\phi\partial_{\theta} -\cot\theta\cos\phi\partial_{\phi},\\
\Omega_3 &= \cos\phi\partial_{\theta} -\cot\theta\sin\phi\partial_{\phi},
\end{align*}
along with the constants of motion $e, l_1,$ and the composite $q = l_2^2 + l_3^2$
\begin{align*}
e &= -(1-\mu)\dot{\gamma}^{t},\\
l_1 &= r^2\sin^2\theta \dot{\gamma}^{\phi},\\
q &= (r^2\dot{\gamma}^{\theta})^2 + \cot^2\theta(r^2\sin^2\theta \dot{\gamma}^{\phi})^2.
\end{align*}

Substituting these three constants, the null geodesic condition
\[0 = g_{ab}\dot{\gamma}^{a}\dot{\gamma}^{b} = -(1-\mu)(\dot{\gamma}^{t})^2 + (1-\mu)^{-1}(\dot{\gamma}^{r})^2 + r^2(\dot{\gamma}^{\theta})^2 + r^2\sin^2\theta(\dot{\gamma}^{\phi})^2\]
gives rise to a simple radial equation
\[r^4(\dot{\gamma}^{r})^2 = r^4e^2 - r^2(1-\mu)(q + l_1^2) =: \mathcal{R}(r,e,q,l_1).\]
Stationary solutions (i.e. trapped null geodesics) are the solution set of the equations
\begin{align*}
\mathcal{R} &= r^4e^2 -r^2(1-\mu)(q+l_1^2) = 0,\\
\partial_{r}\mathcal{R} &= 4r^3e^2 -\left(2r(1-\mu)-r^2\mu_{r}\right)(q+l_1^2) = 0.
\end{align*}
Combining the two, we obtain a linear equation
\[\left[2r(1-\mu) + r^2\mu_{r}\right] (q + l_1^2) = \left[2r-6M\right](q + l_1^2)=0,\]
vanishing at $r_{trapped} = 3M$.  Hence trapping occurs at the photon sphere $r = 3M$, regardless of our choice of cosmological constant $\Lambda$.

\subsection{Sphere Bundles}

Throughout this work, we consider quantities which are scalars and co-vectors on the spheres of symmetry.  The associated sphere bundles, respectively referred to as $\mathcal{L}(0)$ and $\mathcal{L}(-1),$ come equipped with projected covariant derivative operators $\slashed{\nabla}$, defined for scalars by ordinary differentiation and for co-vectors by
\begin{equation}\label{projCov}
\slashed{\nabla}_a dx^A =-\Gamma^{A}_{aB} dx^B.
\end{equation}

Given a spherical co-vector $\omega$, i.e. a section of $\mathcal{L}(-1)$, we define divergence and curl operators by
\begin{align}
\begin{split}
\slashed{div}\ \omega := \slashed{g}^{AB}\slashed{\nabla}_{A}\omega_{B},\\
\slashed{curl}\ \omega := \slashed{\epsilon}^{AB}\slashed{\nabla}_{A}\omega_{B},
\end{split}
\end{align}
and the tensorial spherical Laplacian by
\begin{equation}
\slashed{\Delta}_{\mathcal{L}(-1)} := \slashed{\nabla}^{A}\slashed{\nabla}_{A},
\end{equation}
extending the scalar spherical Laplacian.

In addition to the spherical operators above, we shall make use of spacetime d'Alembertian operators, defined by\begin{equation}
\slashed{\Box}_{\mathcal{L}(-s)} := \slashed{\nabla}^a\slashed{\nabla}_a,
\end{equation}
with $s = 0, 1$ and the appropriate covariant derivative operator.  Note that $\slashed{\Box}_{\mathcal{L}(0)} = \Box$ is the standard d'Alembertian operator on the Schwarzschild-de Sitter spacetime.

\section{The Maxwell Equations on Schwarzschild-de Sitter}

An alternating two-form $F \in \Lambda^2(\mathcal{M})$ is a solution of the Maxwell equations on $(\mathcal{M},g)$ if $F$ satisfies
\begin{equation}\label{MaxwellEqns}
\nabla^{a}F_{ab} = 0, \hspace{5mm} \nabla_{[a}F_{bc]} = 0,
\end{equation}
or equivalently,
\begin{equation}\label{MaxwellEqns2}
dF = 0, \hspace{5mm} d \star F = 0.
\end{equation}
We refer to such solutions as Maxwell tensors on $(\mathcal{M},g)$.

The primary purpose of this section is to study the structure of the Maxwell equations, expressed in a double null frame.

\subsection{Null Decomposition of the Maxwell Equations}
Using the null coordinates \eqref{uvDef}, we define null directions
\begin{align}
\begin{split}
L &:= \partial_{v} = \partial_{t} + \partial_{r_{*}},\\
\underline{L} &:= \partial_{u} = \partial_{t} - \partial_{r_{*}},
\end{split}
\end{align}
spanning the normal bundle of the spheres of symmetry $S^2(t,r) = S^2(u,v)$.  We complete our null frame by choosing orthonormal basis vectors $e_{A}, A = 1,2$, for each of the spheres of symmetry.  Note that the pair $(L,\underline{L})$ is irregular at either of the horizons; rescaling, we define the pairs
\begin{align}\label{eventHorizonPair}
\begin{split}
e_3 &:= L,\\
e_4 &:= (1-\mu)^{-1}\underline{L},
\end{split}
\end{align}
regular across $\mathcal{H}^{+}$, and
\begin{align}\label{cosmoHorizonPair}
\begin{split}
\bar{e}_3 &:= (1-\mu)^{-1}L,\\
\bar{e}_4 &:= \underline{L},
\end{split}
\end{align}
regular across $\overline{\mathcal{H}}^{+}.$

With the null frame $\{ L, \underline{L}, e_1, e_2 \}$ in hand, we decompose the Maxwell tensor into the components
\begin{align}\label{nullMaxwellTensor}
\begin{split}
\alpha_{A} &:= F(e_{A},L),\\
\underline{\alpha}_{A} &:= F(e_{A},\underline{L}),\\
\rho &:= \frac{1}{2}(1-\mu)^{-1}F(L,\underline{L}),\\
\sigma &:= \frac{1}{2}\slashed{\epsilon}^{CD}F_{CD},
\end{split}
\end{align}
with $\alpha_{A}$ and $\underline{\alpha}_{A}$ regarded as one-forms on the spheres of symmetry, i.e. sections of $\mathcal{L}(-1)$, and $\rho$ and $\sigma$ regarded as functions on the same.  

\begin{proposition}\label{nullMaxwell}
Expressed in terms of the null decomposition above, the Maxwell equations \eqref{MaxwellEqns} take the form

\begin{align}
\label{MaxwellOne}\frac{1}{r}\slashed{\nabla}_{L}(r\underline{\alpha}_{A}) + (1-\mu)\left(\slashed{\nabla}_{A}\rho - \slashed{\epsilon}_{AB}\slashed{\nabla}^{B}\sigma\right) &= 0,\\
\label{MaxwellTwo}\frac{1}{r}\slashed{\nabla}_{\underline{L}}(r\alpha_{A}) - (1-\mu)\left(\slashed{\nabla}_{A}\rho + \slashed{\epsilon}_{AB}\slashed{\nabla}^{B}\sigma\right) &= 0,\\
\label{MaxwellThree}\slashed{curl}\ \underline{\alpha} - 2\frac{1-\mu}{r}\sigma + \slashed{\nabla}_{\underline{L}}\sigma &= 0,\\
\label{MaxwellFour}-\slashed{div}\ \underline{\alpha} +2\frac{1-\mu}{r}\rho - \slashed{\nabla}_{\underline{L}}\rho &= 0,\\
\label{MaxwellFive}\slashed{curl}\ \alpha + 2\frac{1-\mu}{r}\sigma + \slashed{\nabla}_{L}\sigma &= 0,\\
\label{MaxwellSix}\slashed{div}\ \alpha - 2\frac{1-\mu}{r}\rho - \slashed{\nabla}_{L}\rho &= 0.
\end{align}

\end{proposition}

For a thorough derivation of the equations above, we refer the reader to Pasqualotto \cite{Pasqualotto}.

\subsection{Coulomb Solutions}
The Maxwell equations \eqref{MaxwellEqns} possess well-known stationary solutions, referred to as Coulomb solutions.  Concretely, given real constants $B$ and $E$, two-tensors of the form
\begin{equation}
F_{ab}dx^{a}\wedge dx^{b} = Br^{-2}\slashed{\epsilon}_{AB}dx^{A}\wedge dx^{B} + Er^{-2}(1-\mu)dt\wedge dr_{*}
\end{equation}
form a two-parameter family of stationary solutions, referred to as Coulomb solutions, to \eqref{MaxwellEqns}.  In terms of the null decomposition above, Coulomb solutions take the form
\begin{align}\label{CoulombNull}
\begin{split}
&\alpha_{A} = \underline{\alpha}_{A} = 0,\\
& \rho = Er^{-2},\\
& \sigma = Br^{-2}.
\end{split}
\end{align}

The main theorem of this work concerns decay of a general solution of the Maxwell equations, specified by appropriate initial data, to a Coulomb solution.  Equivalently, utilizing initial data to identify the asymptotic Coulomb solution, we can reformulate our result in terms of decay of normalized solutions to zero.  We describe this procedure below.

Integrating \eqref{MaxwellThree} and \eqref{MaxwellFive} over the unit sphere, we find the relations
\begin{align*}
\int_{S^2}\partial_{u}(r^2\sigma) &= 0,\\
\int_{S^2}\partial_{v}(r^2\sigma) &=0,
\end{align*}
with similar relations for $\rho$ following from \eqref{MaxwellFour} and \eqref{MaxwellSix}.  That is, we have conservation of the integral quantities, 
\begin{align}
\begin{split}
\int_{S^2(u,v)} \rho &= \int_{S^2(\tilde{u},\tilde{v})} \rho,\\
\int_{S^2(u,v)} \sigma &= \int_{S^2(\tilde{u},\tilde{v})} \sigma,
\end{split}
\end{align}
for general solutions to the Maxwell equations.  This phenomenon is often referred to as conservation of charge; see \cite{Andersson} for an excellent discussion.

The conservation of charge above allows us to identify the asymptotic Coulomb solution, owing to the preservation of its parameters $B$ and $E$.  Given initial data on $C_{u_0,v_0}$ \eqref{nullHypersurfaces} for a Maxwell tensor $F$, we identify its Coulomb parameters by integrating
\begin{align*}
E &= \frac{1}{4\pi}\int_{S^2(u_0,v_0)} \rho,\\
B &= \frac{1}{4\pi}\int_{S^2(u_0,v_0)} \sigma.
\end{align*}
Indeed, integration over any sphere of symmetry lying in $C_{u_0,v_0}$ yields the parameters.  We denote the associated Coulomb solution by $F_{stationary}$, with null decomposition
\begin{align}\label{Fstationary}
\begin{split}
&\alpha_{A} = \underline{\alpha}_{A} = 0,\\
&\bar\rho = \frac{1}{4\pi r(u,v)^2}\int_{S^2(u,v)} \rho,\\
&\bar\sigma = \frac{1}{4\pi r(u,v)^2}\int_{S^2(u,v)} \sigma,
\end{split}
\end{align}
where we have utilized conservation of charge.  Subtracting the associated initial data, we can form a normalized initial data set, with associated normalized solution $F - F_{stationary}$, such that decay of the normalized solution to zero is equivalent to decay of the general solution $F$ to the Coulomb solution $F_{stationary}$.  We remark that this normalization procedure is possible for a variety of initial data specifications, beyond the null hypersurfaces $C_{u_0,v_0}$.

\subsection{The Spin $\pm$1 Teukolsky Equations}
The decoupling of $\alpha_{A}$ and $\underline{\alpha}_{A}$ in the null decomposed Maxwell system of Proposition \ref{nullMaxwell} was established by Teukolsky \cite{Teukolsky} on vacuum, Petrov type-D backgrounds by means of certain algebraic and differential manipulations.  The procedure does not use the vacuum assumption in any meaningful way, so there is a straightforward extension to our setting:

\begin{lemma}
The Maxwell components $\alpha_{A}$ and $\underline{\alpha}_{A}$ satisfy the spin $\pm$1 Teukolsky equations
\begin{align}\label{TeukolskyEqn1}
\begin{split}
&\slashed{\nabla}_{L}\slashed{\nabla}_{\underline{L}}(r\alpha_{A}) + \frac{2}{r}\left(1-\frac{3M}{r}\right)\slashed{\nabla}_{\underline{L}}(r\alpha_{A})\\
&- (1-\mu)\slashed{\Delta}_{\mathcal{L}(-1)}(r\alpha_{A}) + \frac{1-\mu}{r^2}(r\alpha_{A}) = 0,
\end{split}
\end{align}
\begin{align}\label{TeukolskyEqn2}
\begin{split}
&\slashed{\nabla}_{\underline{L}}\slashed{\nabla}_{L}(r\underline{\alpha}_{A}) - \frac{2}{r}\left(1-\frac{3M}{r}\right)\slashed{\nabla}_{L}(r\underline{\alpha}_{A})\\
&- (1-\mu)\slashed{\Delta}_{\mathcal{L}(-1)}(r\underline{\alpha}_{A}) + \frac{1-\mu}{r^2}(r\underline{\alpha}_{A}) = 0.
\end{split}
\end{align}
\end{lemma}

\begin{proof}
We derive the decoupled equation for $\underline{\alpha}_{A}$, that for $\alpha_{A}$ being analogous.  At the outset, we note that
\begin{equation}
\slashed{\nabla}_{L} \slashed{\epsilon}_{AB} = 0, \hspace{5mm} \slashed{\nabla}_{L}\slashed{g}_{AB} = 0, \hspace{5mm} [r\slashed{\nabla}_{A},\slashed{\nabla}_{L}] = 0,
\end{equation}
with similar statements holding for $\underline{L}$.  

Multiplying \eqref{MaxwellOne} by $r$ and applying the operator $\slashed{\nabla}_{\underline{L}}$ to the result, we deduce
\[ \slashed{\nabla}_{\underline{L}}\slashed{\nabla}_{L}(r\underline{\alpha}_{A}) +(1-\mu)\mu_{r}\left(r(\slashed{\nabla}_{A}\rho - \slashed{\epsilon}_{AB}\slashed{\nabla}^{B}\sigma)\right) + (1-\mu)r\left(\slashed{\nabla}_{A}\slashed{\nabla}_{\underline{L}}\rho - \slashed{\epsilon}_{AB}\slashed{\nabla}^{B}\slashed{\nabla}_{\underline{L}}\sigma\right) = 0.\]

Rewriting the last term with \eqref{MaxwellThree} and \eqref{MaxwellFour}, we find
\begin{align*}
&\slashed{\nabla}_{\underline{L}}\slashed{\nabla}_{L}(r\underline{\alpha}_{A}) +(1-\mu)\mu_{r}\left(r(\slashed{\nabla}_{A}\rho - \slashed{\epsilon}_{AB}\slashed{\nabla}^{B}\sigma)\right) + 2(1-\mu)^2\left(\slashed{\nabla}_{A}\rho - \slashed{\epsilon}_{AB}\slashed{\nabla}^{B}\sigma\right)\\
&+ (1-\mu)r\left(-\slashed{\nabla}_{A}\slashed{div}\ \underline{\alpha} + \slashed{\epsilon}_{AB}\slashed{\nabla}^{B}\slashed{curl}\ \underline{\alpha}\right) = 0.
\end{align*}

Application of \eqref{MaxwellOne} and the relation
\[ \slashed{\nabla}_{A}\slashed{div}\ \omega - \slashed{\epsilon}_{AB}\slashed{\nabla}^{B}\slashed{curl}\ \omega = \slashed{\Delta}_{\mathcal{L}(-1)}\omega_{A} - \frac{1}{r^2}\omega_{A},\]
which holds for spherical one-forms $\omega$, yield the spin -1 Teukolsky equation for $\underline{\alpha}$:
\begin{align*}
&\slashed{\nabla}_{\underline{L}}\slashed{\nabla}_{L}(r\underline{\alpha}_{A}) -\left(\mu_{r} + 2\frac{1-\mu}{r}\right)\slashed{\nabla}_{L}(r\underline{\alpha}_{A}) - (1-\mu)\slashed{\Delta}_{\mathcal{L}(-1)}(r\underline{\alpha}_{A}) + \frac{1-\mu}{r^2}(r\underline{\alpha}_{A}) = 0,\\
&\slashed{\nabla}_{\underline{L}}\slashed{\nabla}_{L}(r\underline{\alpha}_{A}) - \frac{2}{r}\left(1-\frac{3M}{r}\right)\slashed{\nabla}_{L}(r\underline{\alpha}_{A}) - (1-\mu)\slashed{\Delta}_{\mathcal{L}(-1)}(r\underline{\alpha}_{A}) + \frac{1-\mu}{r^2}(r\underline{\alpha}_{A}) = 0.
\end{align*}

\end{proof}

\subsection{The Transformation Theory}
Lacking Lagrangian structure, the spin $\pm$1 Teukolsky equations prove difficult to estimate using standard vector field multiplier methods.  However, certain higher order quantities, obtained from $\alpha_{A}$ and $\underline{\alpha}_{A}$ by differential transformations, satisfy equations equipped with such structure and, moreover, having favorable analytic content.  As the starting point for controlling the Maxwell tensor, these quantities are essential to our analysis.

Before proceeding, we remark that, as with the decoupling of the previous subsection, such a transformation theory is well-known on vacuum, Petrov type-D backgrounds (see Chandrasekhar \cite{Chandra1}, Wald \cite{Wald}, and the later work of Aksteiner-B\"{a}ckdahl \cite{Aksteiner}).  The extension to non-vacuum settings, as in the present case and e.g.  \cite{Wald2, Araneda}, appears to be less developed.

We define $P_{A}$ and $\underline{P}_{A}$, each a section of $\mathcal{L}(-1)$, in terms of $\alpha_{A}$ and $\underline{\alpha}_{A}$ as follows:
\begin{align}
P_{A} &:= \frac{r}{1-\mu}\slashed{\nabla}_{\underline{L}}(r\alpha_{A}),\label{P}\\
\underline{P}_{A} &:= \frac{r}{1-\mu}\slashed{\nabla}_{L}(r\underline{\alpha}_{A}) \label{underlineP}.
\end{align}
Observe that both $P_{A}$ and $\underline{P}_{A}$ are regular at the horizons.

\begin{lemma}
The quantities $P_{A}$ and $\underline{P}_{A}$ satisfy the Fackerell-Ipser equation
\begin{align}\label{FIeqn}
\begin{split}
&\slashed{\Box}_{\mathcal{L}(-1)} P_{A} = V P_{A},\\
&\slashed{\Box}_{\mathcal{L}(-1)} \underline{P}_{A} = V \underline{P}_{A},
\end{split}
\end{align}
with $V = \frac{1}{r^2}(1-\mu) + \Lambda$.
\end{lemma}

\begin{proof}
We present the argument for $\underline{P}_{A}$, that for $P_{A}$ being analogous.  Note that the corresponding Teukolsky equation \eqref{TeukolskyEqn2} can be rewritten as 
\[\frac{1-\mu}{r^2}\slashed{\nabla}_{\underline{L}}\left(\frac{r^2}{1-\mu}\slashed{\nabla}_{L}(r\underline{\alpha}_{A})\right) - (1-\mu)\slashed{\Delta}_{\mathcal{L}(-1)}(r\underline{\alpha}_{A}) + \frac{1-\mu}{r^2}(r\underline{\alpha}_{A}) = 0.\]
Multiplying the equation by $\frac{r^2}{1-\mu}$ and applying the operator $\slashed{\nabla}_{L}$, we find
\[\slashed{\nabla}_{L}\slashed{\nabla}_{\underline{L}}\left(\frac{r^2}{1-\mu}\slashed{\nabla}_{L}(r\underline{\alpha}_{A})\right) - \slashed{\nabla}_{L}\left(r^2\slashed{\Delta}_{\mathcal{L}(-1)}(r\underline{\alpha}_{A})\right) + \slashed{\nabla}_{L}(r\underline{\alpha}_{A}) = 0,\]
or
\[\slashed{\nabla}_{L}\slashed{\nabla}_{\underline{L}}\left(\frac{r^2}{1-\mu}\slashed{\nabla}_{L}(r\underline{\alpha}_{A})\right) - r^2\slashed{\Delta}_{\mathcal{L}(-1)}\left(\slashed{\nabla}_{L}(r\underline{\alpha}_{A})\right) + \slashed{\nabla}_{L}(r\underline{\alpha}_{A}) = 0,\]
using the commutation relation $[\slashed{\nabla}_{L},r^2\slashed{\Delta}_{\mathcal{L}(-1)}] = 0$.

Introducing the quantity
\[\underline{\phi}_{A} := \frac{r^2}{1-\mu}\slashed{\nabla}_{L}(r\underline{\alpha}_{A}),\]
we rewrite the expression in terms of $\underline{\phi}_{A}$:
\[\slashed{\nabla}_{L}\slashed{\nabla}_{\underline{L}}\underline{\phi}_{A} - (1-\mu)\slashed{\Delta}_{\mathcal{L}(-1)}\underline{\phi}_{A} + \frac{1-\mu}{r^2}\underline{\phi}_{A} = 0.\]

The above is the form seen in Dafermos-Holzegel-Rodnianski \cite{DHR} and Pasqualotto \cite{Pasqualotto}.  With the rescaling $r\underline{P}_{A} = \underline{\phi}_{A}$, the quantity $\underline{P}_{A}$ is found to satisfy the Fackerell-Ipser equation \eqref{FIeqn} as claimed.
\end{proof}

\section{Analysis of the Fackerell-Ipser Equation}

In this section, we analyze co-vectors $\Psi_{A}$ satisfying the Fackerell-Ipser equation \eqref{FIeqn}.  In particular, the estimates derived in this section hold for $P_{A}$ and $\underline{P}_{A}$.  The analysis is largely based upon the ideas and notation of \cite{DRdS, HKW}.

\subsection{Poincar\'{e} Inequality}
 
For the bundle $\mathcal{L}(-1)$ of spherical co-vectors, the spectrum of the associated spherical Laplacian consists of eigenvalues 
\begin{equation}
\lambda_{m,\ell} = \frac{(1 - \ell(\ell + 1))}{r^2},
\end{equation}
with $\ell \geq 1$ and $|m| \leq \ell$.  The identity
\begin{equation}
\slashed{\Delta} |\omega|^2 = 2\slashed{\Delta}_{\mathcal{L}(-1)}\omega \cdot \omega + 2 |\slashed\nabla \omega|^2
\end{equation}
yields the Poincar\'{e} inequality
\begin{equation}\label{PoincarePsi2}
\int_{S^2(t,r)} |\slashed\nabla\omega|^2 \geq \frac{1}{r^2}\int_{S^2(t,r)}|\omega|^2.
\end{equation}

Here we have used the notation
\begin{equation}\label{angularGradient}
|\slashed\nabla\omega|^2= g^{AB}g^{CD}(\slashed\nabla_{A}\omega)_{C}(\slashed\nabla_{B}\omega)_{D}
\end{equation}
for the angular gradient.

\subsection{Stress-Energy Formalism}
Associated with our Fackerell-Ipser equation is the stress-energy tensor
\begin{equation}\label{stressTensor}
T_{ab}[\Psi] := \slashed{\nabla}_a \Psi\cdot \slashed{\nabla}_b \Psi - \frac{1}{2}g_{ab}(\slashed{\nabla}^{c}\Psi\cdot\slashed{\nabla}_{c}\Psi + V|\Psi|^2),
\end{equation}
where we emphasize that
\begin{align*}
|\Psi|^2 &= g^{AB}\Psi_{A}\Psi_{B},\\
\slashed{\nabla}_{a}\Psi\cdot\slashed{\nabla}_{b}\Psi &= g^{AB}(\slashed\nabla_{a}\Psi)_{A}(\slashed\nabla_{b}\Psi)_{B},\\
\slashed{\nabla}^{c}\Psi\cdot\slashed{\nabla}_{c}\Psi &= g^{ab}g^{AB}(\slashed\nabla_{a}\Psi)_{A}(\slashed\nabla_{b}\Psi)_{B}.
\end{align*}

Applying a vector-field multiplier $X^{b}$, we define the energy current
\begin{equation}
J^{X}_{a}[\Psi] := T_{ab}[\Psi]X^{b}
\end{equation}  
and the density
\begin{equation}
K^{X}[\Psi] := \nabla^{a}J^{X}_{a}[\Psi] = \nabla^{a}(T_{ab}[\Psi]X^{b}).
\end{equation}

As well, we will have occasion to use the weighted energy current
\begin{equation}\label{Jweight}
J^{X,\omega^{X}}_{a}[\Psi] := J^{X}_{a}[\Psi] + \frac{1}{4}\omega^{X}\nabla_{a}|\Psi|^2 -\frac{1}{4}\nabla_{a}\omega^{X}|\Psi|^2, 
\end{equation}
with weighted density
\begin{equation}\label{Kweight}
K^{X,\omega^{X}}[\Psi] := K^{X}[\Psi] + \frac{1}{4}\omega^{X}\Box |\Psi|^2 -\frac{1}{4} \Box \omega^{X} |\Psi|^2,
\end{equation}
for a suitable scalar weight function $\omega^{X}$.

The current $J^{X}_{a}[\Psi]$ and density $K^{X}[\Psi]$ serve as a convenient notation to express the spacetime Stokes' theorem
\begin{equation}\label{stokes}
\int_{\mathcal{\partial D}} J^{X}_{a}[\Psi]\eta^{a} = \int_{\mathcal{D}} K^{X}[\Psi],
\end{equation} 
integrated over a spacetime region $\mathcal{D}$ with boundary $\partial\mathcal{D}$.

Likewise, the weighted quantities satisfy
\begin{equation}
\int_{\mathcal{\partial D}} J^{X,\omega^{X}}_{a}[\Psi]\eta^{a} = \int_{\mathcal{D}} K^{X,\omega^{X}}[\Psi].
\end{equation} 

The stress-energy tensor $T_{ab}[\Psi]$ defined above has non-trivial divergence
\begin{equation}\label{stressDivergence}
\nabla^{a}T_{ab}[\Psi] = -\frac{1}{2}\nabla_{b}V|\Psi|^2 + \slashed{\nabla}^{a}\Psi[\slashed{\nabla}_{a},\slashed{\nabla}_{b}]\Psi,
\end{equation}
where we note that the commutator $[\slashed{\nabla}_{a},\slashed{\nabla}_{b}]$ vanishes when contracted with a multiplier invariant under the angular Killing fields in $\Omega$.  In particular, all such multipliers considered in the subsequent analysis have this property.

We remark that, owing to the positivity of the potential term $V$, the stress-energy tensor satisfies a positive energy condition.  Namely, given future-directed, timelike vector fields $X_1$ and $X_2$, we have
\begin{equation}\label{energyCondition}
T_{ab}[\Psi]X_1^{a}X_2^{b} \geq 0.
\end{equation}

\subsection{Additional Notation}

Our estimates are expressed in terms of the null hypersurfaces $C_{\tau,\tau}$ \eqref{nullHypersurfaces}.  For simplicity, we denote the hypersurfaces
\begin{equation}\label{decayFoliation}
\Sigma_{\tau} := C_{\tau,\tau} = \underline{C}_{\tau,\tau} \cup \overline{C}_{\tau,\tau}
\end{equation}
and the spacetime region
\begin{equation}
\mathcal{R}(\tau',\tau) := J^{+}(\Sigma_{\tau'}) \cap J^{-}(\Sigma_{\tau}),
\end{equation}
where $0 \leq \tau' < \tau$.

Expressed in the Eddington-Finkelstein coordinates, the relevant volume forms are written
\begin{align}\label{volumeForms}
\begin{split}
dVol_{\underline{C}_{\tau,\tau}} &= (1-\mu)r^2\sin\theta dud\theta d\phi,\\
dVol_{\overline{C}_{\tau,\tau}} &= (1-\mu)r^2\sin\theta dvd\theta d\phi,\\
dVol_{\mathcal{R}(\tau',\tau)} &= 4(1-\mu)r^2\sin\theta dudvd\theta d\phi.
\end{split}
\end{align}

In addition, we define the boundary regions
\begin{align}
\begin{split}
&\mathcal{H}^{+}(\tau',\tau) := \mathcal{H}^{+}\cap \mathcal{R}(\tau',\tau),\\
&\overline{\mathcal{H}}^{+}(\tau',\tau) := \overline{\mathcal{H}}^{+}\cap \mathcal{R}(\tau',\tau).
\end{split}
\end{align}

We remark that the specifications above are made for the sake of convenience; the subsequent decay estimates can be expressed with respect to a broad class of foliations.  

Throughout the remainder of this work, we use $c(M,\Lambda)$ and $C(M,\Lambda)$ to denote small and large positive constants, respectively, each depending upon the parameters $M$ and $\Lambda$.

\subsection{The Killing Multiplier $T$}
Applying the static Killing field $T$ as a multiplier, we observe that the density $K^{T}[\Psi]$ vanishes in consequence of $V$ being radial, such that $(\nabla^{a}T_{ab}[\Psi])T^{b}$ vanishes \eqref{stressDivergence}, and $T$ being Killing, such that $\pi_{T}^{ab} = \nabla^{(a}T^{b)}$ vanishes.

Integrating over $\mathcal{R}(\tau',\tau)$, we obtain the identity \eqref{stokes}
\[ \int_{\Sigma_{\tau}} J^{T}_{a}[\Psi]\eta^{a} +\int_{\mathcal{H}^{+}(\tau',\tau)} J^{T}_{a}[\Psi]\eta^{a} +\int_{\overline{\mathcal{H}}^{+}(\tau',\tau)} J^{T}_{a}[\Psi]\eta^{a} = \int_{\Sigma_{\tau'}} J^{T}_{a}[\Psi]\eta^{a}.\]

Defining the $T$-energy by
\begin{equation}\label{Tenergy}
E^{T}_{\Psi}(\tau) := \int_{\Sigma_{\tau}} J^{T}_{a}[\Psi]\eta^{a},
\end{equation}
the identity above yields the estimates
\begin{align}\label{Testimates}
\begin{split}
E_{\Psi}^{T}(\tau) &\leq E_{\Psi}^{T}(\tau'),\\
\int_{\mathcal{H}^{+}(\tau',\tau)} J^{T}_{a}[\Psi]\eta^{a} &\leq E_{\Psi}^{T}(\tau'),\\
\int_{\overline{\mathcal{H}}^{+}(\tau',\tau)} J^{T}_{a}[\Psi]\eta^{a} &\leq E_{\Psi}^{T}(\tau'),
\end{split}
\end{align}
for all $0 \leq \tau' < \tau$. 

With our energy condition \eqref{energyCondition}, we note that the $T$-energy above is non-negative, degenerating at each of the horizons.

\subsection{The Red-Shift Multiplier $N$}

The static Killing field $T$ degenerates at each horizon, becoming null; consequently, the $T$-energy defined above is degenerate and unsuited for proving boundedness and decay results up to and including the horizons.  To circumvent this, we utilize a red-shift multiplier, of the sort introduced in \cite{DR}.  We recall the details below.

We work on the event horizon $\mathcal{H}^{+}$, away from the bifurcation sphere.  Letting $Y$ be a null vector transversal to the Killing field $T$, itself tangential on $\mathcal{H}^{+}$, we specify $Y$ by
\begin{enumerate}
\item $Y$ is future-directed, with normalization $g_{M,\Lambda}(Y,T) = -2$,
\item $Y$ is invariant under $T$ and the $\Omega_{i}$,
\item On $\mathcal{H}^{+}$, $\nabla_{Y}Y = -\sigma(Y + T)$, for $\sigma \in \R$ as yet unchosen.
\end{enumerate}

Taking $e_{A}$ as orthonormal basis vectors, tangential to the spheres of symmetry, we calculate in the normalized null frame $\{ T, Y, e_1, e_2 \}$:
\begin{align}\label{Ndeformation}
\begin{split}
\nabla_{T}Y &=  -\kappa Y,\\
\nabla_{Y}Y &=  -\sigma(T + Y),\\
\nabla_{e_{A}}Y &=  h^{B}_{A}e_{B},
\end{split}
\end{align}
with $\kappa(M,\Lambda)$ being the positive surface gravity on Schwarzschild-de Sitter spacetime, and with $h_{AB}$ being the second fundamental form of the round sphere of radius $r = r_{b}$ (recall that we work on the event horizon) with respect to $Y$.
 
We compute
\begin{align*}
K^{Y}[\Psi] &= (\nabla^{a}T_{ab}[\Psi])Y^{b} + T_{ab}[\Psi]\nabla^{a}Y^{b}\\
&= -\frac{1}{2}Y^{a}\nabla_{a}V|\Psi|^2 + T_{ab}[\Psi]\nabla^{a}Y^{b}.
\end{align*}

As the potential $V$ is increasing near the event horizon, the first term above is non-negative.  Expanding the second term with \eqref{Ndeformation}, we find
\begin{align*}
T_{ab}[\Psi]\nabla^{a}Y^{b} &= \frac{\kappa}{2}T_{ab}[\Psi]Y^{a}Y^{b} + \frac{\sigma}{4}T_{ab}[\Psi]T^{a}T^{b} + \frac{\sigma}{2}T_{ab}[\Psi]T^{a}Y^{b}  +T_{ab}[\Psi]e_{A}^{a}e_{B}^{b}h^{AB}\\
&\geq \frac{\kappa}{2}|\slashed{\nabla}_{Y}\Psi|^2 + \frac{\sigma}{4}|\slashed{\nabla}_{T}\Psi|^2 + \frac{\sigma}{2}(|\slashed{\nabla}\Psi|^2 + V|\Psi|^2) \\
&- c(M,\Lambda)\left(\slashed{\nabla}_{T}\Psi \cdot \slashed{\nabla}_{Y}\Psi + |\slashed{\nabla}\Psi|^2 + V|\Psi|^2\right).
\end{align*}

With positive surface gravity $\kappa$ and a choice of large $\sigma$, we deduce 
\begin{align}
\begin{split}
T_{ab}[\Psi]\nabla^{a}Y^{b} &\geq c(M,\Lambda)\left(|\slashed{\nabla}_{Y}\Psi|^2 + |\slashed{\nabla}_{T}\Psi|^2 + |\slashed{\nabla}\Psi|^2 + V|\Psi|^2\right) \\
&\geq c(M,\Lambda)T_{ab}[\Psi](T+Y)^{a}(T+Y)^{b}.
\end{split}
\end{align}

Together with the positivity of the first density term, we have the estimate
\begin{equation}
K^{T + Y}[\Psi] = K^{Y}[\Psi] \geq c(M,\Lambda) T_{ab}[\Psi](T+Y)^{a}(T+Y)^{b}
\end{equation}
on the event horizon $\mathcal{H}^{+}.$  A similar argument can be made on the cosmological horizon $\overline{\mathcal{H}}^{+}$ using the transversal field $\overline{Y}$ and $T$.

Extending to the static region, we construct a strictly timelike red-shift multiplier, identically $N = T + Y$ on $\mathcal{H}^{+}$ and $N = T + \overline{Y}$ on $\overline{\mathcal{H}}^{+}$, satisfying the estimates
\begin{align}\label{Nestimates}
\begin{split}
K^{N}[\Psi] &\geq{c(M,\Lambda) J^{N}_{a}[\Psi]N^{a}}\hspace{6mm} \textup{for}\ r_{b}\leq r \leq r_1\ \textup{and}\ R_1 \leq r \leq r_{c},\\
J_{a}^{N}[\Psi]T^{a} &\sim J_{a}^{T}[\Psi]T^{a}\hspace{20mm} \textup{for}\ r_1 \leq r \leq r_2\ \textup{and}\ R_2 \leq r \leq R_1,\\
|K^{N}[\Psi]| &\leq  C(M,\Lambda)|J_{a}^{T}[\Psi]T^{a}|\hspace{3.5mm} \textup{for}\ r_1 \leq r \leq r_2\ \textup{and}\ R_2 \leq r \leq R_1,\\
N &= T\hspace{32mm} \textup{for}\ r_2 \leq r \leq R_2,
\end{split}
\end{align}
for radii $r_{b} < r_1 < r_2 < R_2 < R_1 < r_{c}$.  We choose the radii such that
\[ -\infty < (r_1)_* < (r_2)_* < 0 <  (R_2)_* < (R_1)_* < \infty,\]
expressed in the Regge-Wheeler coordinate.  That is, the radii are well separated from one another, and moreover, the red-shift vector $N$ is identically $T$ in a region about the photon sphere $(3M)_* = 0$.

As an immediate application, we prove uniform boundedness of the non-degenerate $N$-energy for solutions to the Fackerell-Ipser equation.  Here, the $N$-energy is defined by
\begin{equation}
E^{N}_{\Psi}(\tau) := \int_{\Sigma_{\tau}} J^{N}_{a}[\Psi]\eta^{a}.
\end{equation}

\begin{theorem}\label{FIboundedness}
Suppose $\Psi$ is a solution of \eqref{FIeqn}, specified by smooth initial data on the hypersurface $\Sigma_{0}$.   Then for $\tau > \tau' \geq 0$, $\Psi$ satisfies the uniform energy estimate
\begin{equation}
E^{N}_{\Psi}(\tau) \leq C(M,\Lambda)E^{N}_{\Psi}(\tau').
\end{equation}
\end{theorem}

\begin{proof}
The proof proceeds just as in \cite{DRClay}.  Given $\tau' \leq \tilde{\tau} \leq \tau$, integration over the spacetime region $\mathcal{R}(\tilde{\tau}, \tau)$ yields
 \begin{align*}
&E^{N}_{\Psi}(\tau) + \int_{\mathcal{R}(\tilde{\tau}, \tau)\cap{\left(\{r_b \leq r \leq r_1\}\cup\{R_1\leq r\leq r_c\}\right)}}K^N[\Psi]\\
&\leq E^{N}_{\Psi}(\tilde{\tau})+\int_{\mathcal{R}(\tilde{\tau}, \tau)\cap{\left(\{r_1 \leq r \leq r_2\}\cup\{R_2\leq r\leq R_1\}\right)}}|K^N[\Psi]|,
\end{align*}
as the horizon terms have good sign.

Utilizing monotonicity of the $T$-energy \eqref{Testimates} and the red-shift estimates \eqref{Nestimates}, we deduce the integral inequality
\[E_{\Psi}^{N}(\tau) + c(M,\Lambda)\int_{\tilde{\tau}}^{\tau} E_{\Psi}^{N}(s) ds \leq C(M,\Lambda)E_{\Psi}^{T}(\tau')(\tau - \tilde{\tau}) + E_{\Psi}^{N}(\tilde{\tau}),\]
which implies the uniform bound
\[E_{\Psi}^{N}(\tau) \leq C(M,\Lambda) E_{\Psi}^{N}(\tau').\]
\end{proof}

\subsection{The Morawetz Multiplier $X$}
Let $X = f(r)\partial_{r_{*}}$, with $f$ a radial function, and let $\omega^{X}$ be a scalar weight function.  Using the notation $(\hspace{2mm})'$ to denote differentiation by the Regge-Wheeler coordinate $r_{*}$, we calculate the unweighted density to be
\begin{align}
\begin{split}
K^{X}[\Psi]  &= f'|\slashed\nabla_{r_{*}}\Psi|^2 + \frac{f}{r}\left(1-\frac{3M}{r}\right)|\slashed\nabla\Psi|^2 -\frac{1}{2}\left(f' + \frac{2}{r}(1-\mu)f\right)\slashed{\nabla}^{a}\Psi\cdot\slashed{\nabla}_{a}\Psi\\
&+\left(\frac{1}{2} Vf\mu_{r} -\frac{1}{2}fV' -\frac{1}{2}\left(f' + \frac{2}{r}(1-\mu)f\right)V\right)|\Psi|^2\\
&= f'|\slashed\nabla_{r_{*}}\Psi|^2 + \frac{f}{r}\left(1-\frac{3M}{r}\right)|\slashed\nabla\Psi|^2 -\frac{1}{4}\left(f' + \frac{2}{r}(1-\mu)f\right)\Box |\Psi|^2\\
& +\left(\frac{1}{2} Vf\mu_{r} -\frac{1}{2}fV'\right)|\Psi|^2,
\end{split}
\end{align}
where we have used the identity
\begin{equation}\label{productRule}
\Box |\Psi|^2 = 2V|\Psi|^2 + 2\slashed{\nabla}^{a}\Psi\cdot\slashed\nabla_{a}\Psi.
\end{equation}

Inserting the weight function
\begin{equation}
\omega^{X} := f' + \frac{2}{r}(1-\mu)f,
\end{equation}
we calculate the weighted density \eqref{Kweight}
\begin{align}\label{weightedDensity}
\begin{split}
K^{X,\omega^{X}}[\Psi] &= f'|\slashed\nabla_{r_{*}}\Psi|^2 + \frac{f}{r}\left(1 - \frac{3M}{r}\right)|\slashed\nabla \Psi|^2\\
& + \left(\frac{1}{2}Vf\mu_{r} - \frac{1}{2}fV' - \frac{1}{4}\Box\omega^{X}\right)|\Psi|^2.
\end{split}
\end{align}

Through the application of suitable multipliers of the form above, we deduce the following non-degenerate integrated decay estimate:

\begin{theorem}\label{FIdecay}
Suppose $\Psi$ is a solution of \eqref{FIeqn}, specified by smooth initial data on the hypersurface $\Sigma_{0}$.   Then for $\tau \geq \tau' \geq 0$, $\Psi$ satisfies the non-degenerate integrated decay estimate
\begin{equation}
\int_{\tau'}^{\tau} E^{N}_{\Psi}(s)ds \leq C(M,\Lambda)\left(E^{N}_{\Psi}(\tau') + E^{N}_{\Omega\Psi}(\tau')\right).
\end{equation}
\end{theorem}

\begin{proof}
The multiplier $X_1 = f_1(r)\partial_{r_{*}}$, with
\begin{equation}
f_1(r) = \left(1-\frac{3M}{r}\right)\left(1+\frac{\mu}{2}\right)^2,
\end{equation}
provides the primary density estimate
\begin{multline}\label{bulk}
\int_{\mathcal{R}(\tau',\tau)\cap{\{r_1 \leq r \leq R_1\}}} \left[|\Psi|^2 + |\slashed\nabla_{r_{*}}\Psi|^2 + (r-3M)^2|\slashed{\nabla}\Psi|^2\right] \\
\leq C(M,\Lambda)\int_{\mathcal{R}(\tau',\tau)} K^{X_1,\omega^{X_1}}[\Psi],
\end{multline}
giving coercive control away from the horizons and away from the photon sphere $r = 3M$.

With standard modifications by the multiplier $X_2 = r^2\partial_{r_{*}}$ \cite{DRdS}, allowing for control of $t$-derivatives, and the red-shift multiplier $N$, allowing for control near the horizons, we deduce the density estimates
\begin{multline}\label{improvedbulkOne}
\int_{\tau'}^{\tau}\left(\int_{\Sigma_{s}\cap{\{r_1\leq r\leq R_1\}}} J^{T}_{a}[\Psi]\eta^{a}\right)ds\\
\leq C(M,\Lambda)\int_{\mathcal{R}(\tau',\tau)}\Big[K^{X_1,\omega^{X_1}}[\Psi] + K^{X_1,\omega^{X_1}}[\Omega\Psi] + K^{c_1(M,\Lambda)X_2}[\Psi]\Big],
\end{multline}
\begin{align}\label{improvedbulkTwo}
\begin{split}
\int_{\tau'}^{\tau}E^{N}_{\Psi}(s)ds \leq C(M,\Lambda)&\int_{\mathcal{R}(\tau',\tau)}\Big[K^{X_1,\omega^{X_1}}[\Psi] + K^{X_1,\omega^{X_1}}[\Omega\Psi]\\
& + K^{c_1(M,\Lambda)X_2}[\Psi] + K^{c_2(M,\Lambda)N}[\Psi]\Big],
\end{split}
\end{align}
where $c_1(M,\Lambda)$ and $c_2(M,\Lambda)$ are suitably chosen positive constants.

Turning to the boundary terms, we note that those terms formed from the weighted $X_1$-energy are bounded by the $T$-energy:
\begin{align}\label{boundaryPart}
\begin{split}
&\Big{|}\int_{\Sigma_{\tau}} J^{X_1,\omega^{X_1}}_{a}[\Psi]\eta^{a}\Big{|} \leq C(M,\Lambda) E^{T}_{\Psi}(\tau) \leq C(M,\Lambda)E^{T}_{\Psi}(\tau'),\\
&\Big{|}\int_{\mathcal{H}^{+}(\tau',\tau)} J^{X_1,\omega^{X_1}}_{a}[\Psi]\eta^{a}\Big{|} \leq C(M,\Lambda) \int_{\mathcal{H}^{+}(\tau',\tau)} J^{T}_{a}[\Psi]\eta^{a} \leq C(M,\Lambda)E_{\Psi}^{T}(\tau'),\\
&\Big{|}\int_{\overline{\mathcal{H}}^{+}(\tau',\tau)} J^{X_1,\omega^{X_1}}_{a}[\Psi]\eta^{a}\Big{|} \leq C(M,\Lambda)\int_{\overline{\mathcal{H}}^{+}(\tau',\tau)} J^{T}_{a}[\Psi]\eta^{a} \leq C(M,\Lambda)E_{\Psi}^{T}(\tau'),
\end{split}
\end{align}
where we have also utilized \eqref{Testimates}.  Similar estimates hold for the unweighted $X_2$-energy.  Using these boundary estimates and the degenerate density estimate \eqref{improvedbulkOne}, we can estimate the horizon terms formed from the red-shift:
\begin{align}
\begin{split}
&\int_{\mathcal{H}^{+}(\tau',\tau)} J^{N}_{a}[\Psi]\eta^{a} \leq C(M,\Lambda)\left(E^{N}_{\Psi}(\tau') + E^{N}_{\Omega\Psi}(\tau')\right),\\ 
&\int_{\overline{\mathcal{H}}^{+}(\tau',\tau)} J^{N}_{a}[\Psi]\eta^{a} \leq C(M,\Lambda)\left(E^{N}_{\Psi}(\tau') + E^{N}_{\Omega\Psi}(\tau')\right).
\end{split}
\end{align}
Taken together with Theorem \ref{FIboundedness}, these boundary estimates and the non-degenerate density estimate \eqref{improvedbulkTwo} lead to an integrated decay estimate for the $N$-energy:
\begin{equation}
\int_{\tau'}^{\tau} E^{N}_{\Psi}(s)ds \leq C(M,\Lambda)\left(E^{N}_{\Psi}(\tau') + E^{N}_{\Omega\Psi}(\tau')\right).
\end{equation}
\end{proof}
\subsection{Decay Estimates}

Using Theorems \ref{FIboundedness} and \ref{FIdecay}, we conclude this section with decay estimates on the solution $\Psi$.

\begin{theorem}\label{degExpDecay}
Suppose $\Psi_{\ell}$ is a solution of \eqref{FIeqn}, specified by smooth initial data on the hypersurface $\Sigma_{0}$ and supported at the harmonic $\ell \geq 1$.  Then for $\tau \geq 0$, $\Psi_{\ell}$ satisfies the energy decay estimate
\begin{equation}
E^{N}_{\Psi_{\ell}}(\tau) \leq C(M,\Lambda)E^{N}_{\Psi_{\ell}}(0)\exp(-c(M,\Lambda)\tau/\ell^2).
\end{equation}
\end{theorem}

\begin{proof}
Letting $f(\tau) = E^{N}_{\Psi_{\ell}}(\tau)$,  Theorems \ref{FIboundedness} and \ref{FIdecay} imply
\begin{align*}
f(\tau) &\leq C_1(M,\Lambda) f(\tau'),\\
\int_{\tau'}^{\tau} f(s)ds &\leq C_2(M,\Lambda, \ell)f(\tau'),
\end{align*}
where $C_2(M,\Lambda,\ell) = C_3(M,\Lambda)(1+\ell^2)$ has quadratic dependence on $\ell$.  Subsequently, we suppress dependence of the $C_{i}$ on the parameters $M, \Lambda,$ and $\ell$.

Letting $\lambda > 0$ be a positive parameter and integrating over $[\tau, \tau + \lambda C_2],$ the mean value theorem yields a $\tilde{\tau}$ on the interval such that
\[ \lambda C_2 f(\tilde\tau) = \int_{\tau}^{\tau + \lambda C_2} f(s)ds \leq C_2 f(\tau),\]
where we have used the integral estimate.  Applying as well the pointwise estimate, we find
\[\frac{\lambda C_2}{C_1} f(\tau + \lambda C_2) \leq \lambda C_2 f(\tilde\tau) \leq C_2 f(\tau).\]
With the choice of parameter $\lambda = 2C_1$, we have
\[ f(\tau + 2C_1C_2) \leq \frac{1}{2}f(\tau).\]
That is, given any initial point $\tau_0 = \tau \geq 0$, we can produce a sequence $\tau_{k} = \tau_0 + 2kC_1C_2$ exhibiting the exponential decay
\[ f(\tau_{k}) \leq \frac{1}{2^k}f(\tau_0),\]
with decay parameter inversely proportional to the product $C_1C_2$.  In light of the estimates above, this sequential result is easily extended to arbitrary $\tau$, establishing the theorem.

\end{proof}

The degenerate exponential decay estimate above easily leads to non-degenerate super-polynomial decay of general solutions $\Psi$:
\begin{theorem}\label{PolyDecay}
Suppose $\Psi$ is a solution of \eqref{FIeqn}, specified by smooth initial data on the hypersurface $\Sigma_{0}$.  Then for $\tau \geq 0$, $\Psi$ satisfies the energy decay estimate
\begin{equation}
E^{N}_{\Psi}(\tau) \leq C(M,\Lambda,m)\left(\sum_{(q) \leq m} E^{N}_{\Omega^{(q)}\Psi}(0)\right)(1+\tau)^{-m}.
\end{equation}
\end{theorem}
Note that the initial energy involves commutation of $\Psi$ with the angular Killing fields of $\Omega$, these commutations being described by multi-indices $(q)$ of length $m$ or less.

\begin{proof}
Energy decay follows from an application of the comparison
\[ \exp(\tau) \geq \frac{\tau^m}{m!},\]
satisfied for $\tau \geq 0$, to the result of Theorem \ref{degExpDecay}, and from the $L^2$-summability of the co-vector harmonics over the spheres of symmetry in $\Sigma_{\tau}$.
\end{proof}
We remark that pointwise decay follows from further commutation with the angular Killing fields of $\Omega$ and application of standard Sobolev embedding to the resulting energy estimates.

\section{Decay of the Maxwell Components}

Using the estimates on $P_{A}$ and $\underline{P}_{A}$ from the previous section, we conclude this work by proving decay of the components of $F_{ab}$.

\subsection{Decay of $\alpha_{A}$ and $\underline{\alpha}_{A}$}
Basic estimates for $\alpha_{A}$ and $\underline\alpha_{A}$ are obtained by exploiting the transformation formulae used in defining $P_{A}$ \eqref{P} and $\underline{P}_{A}$ \eqref{underlineP}, with higher order statements obtained by means of commutation.  Combining uniform boundedness and integrated decay estimates, we obtain degenerate exponential and non-degenerate super-polynomial decay of $\alpha_{A}$ and $\underline\alpha_{A}$ in much the same way as in Theorems \ref{degExpDecay} and \ref{PolyDecay}.

Our estimates split naturally into three radial regions relating to the critical radius $r_{\mu}$, where $\mu$ is minimized; concretely,
\begin{equation}\label{criticalRadius}
r_{\mu} := \left(\frac{3M}{\Lambda}\right)^{1/3}.
\end{equation}

We let 
\begin{align}
\begin{split}
r_3 &:= \min\{3M, r_{\mu}/2\},\\
R_3 &:= \max\{3M, 3r_{\mu}/2\},
\end{split}
\end{align}
and use the shorthand
\begin{align}
\begin{split}
\RN{1} &:= \{r_{b} \leq r \leq r_3\},\\
\RN{2}&:= \{r_3 \leq r \leq R_3\},\\
\RN{3}&:= \{R_3 \leq r \leq r_{c}\},
\end{split}
\end{align}
noting that $\mu$ is strictly decreasing on $\RN{1}$ and strictly increasing on $\RN{3}$.

We present the analysis for $\underline{\alpha}_{A}$, that for $\alpha_{A}$ being analogous.  As $\underline\alpha$ fails to be regular at the event horizon, we introduce the normalized quantity 
\begin{equation}
\tilde{\underline\alpha}_{A} := (1-\mu)^{-1}\underline\alpha_{A}.
\end{equation}
Throughout, we will make use of the $N$-energy for $\underline\alpha$, specified by
\begin{equation}
E^{N}_{\underline\alpha}(\tau) := \int_{\underline{C}_{\tau,\tau}} J^{N}_{a}[\tilde{\underline\alpha}]\eta^{a} + \int_{\overline{C}_{\tau,\tau}} J^{N}_{a}[\underline\alpha]\eta^{a}.
\end{equation}
We remind the reader of the volume form conventions \eqref{volumeForms}.

\subsubsection{Region $\RN{1}$}

Applying \eqref{underlineP} and H\"{o}lder's inequality, we obtain the differential inequality
\begin{align}\label{differentialOne}
\begin{split}
\slashed\nabla_{L}\left(r^2(1-\mu)|\tilde{\underline\alpha}|^2\right) &= r^2(1-\mu)\mu_{r}|\tilde{\underline\alpha}|^2 + 2(1-\mu)\tilde{\underline\alpha}\cdot\underline{P}\\
&\leq -c(M,\Lambda)(1-\mu)|\tilde{\underline\alpha}|^2 + C(M,\Lambda)(1-\mu)|\underline{P}|^2.
\end{split}
\end{align}
Integrating \eqref{differentialOne} over the spacetime region $\mathcal{R}(\tau',\tau) \cap \RN{1}$ yields the integrated decay estimate
\begin{align}\label{intSmall}
\begin{split}
\int_{\mathcal{R}(\tau',\tau) \cap {\RN{1}}} |\tilde{\underline{\alpha}}|^2
&\leq C(M,\Lambda)\left[\int_{\mathcal{R}(\tau',\tau) \cap {\RN{1}}} |\underline{P}|^2 + \int_{\underline{C}_{\tau',\tau'}\cap {\RN{1}}} |\tilde{\underline\alpha}|^2\right]\\
&\leq C(M,\Lambda)\left[E^{N}_{\underline{P}}(\tau') + E^{N}_{\Omega\underline{P}}(\tau') + E^{N}_{\underline\alpha}(\tau')\right],
\end{split}
\end{align}
and the energy estimate
\begin{align}\label{energySmall}
\begin{split}
\int_{\underline{C}_{\tau,\tau}\cap {\RN{1}}} |\tilde{\underline{\alpha}}|^2
&\leq C(M,\Lambda)\left[\int_{\mathcal{R}(\tau',\tau) \cap{\RN{1}}} |\underline{P}|^2 + \int_{\underline{C}_{\tau',\tau'}\cap {\RN{1}}} |\tilde{\underline\alpha}|^2\right]\\
&\leq C(M,\Lambda)\left[E^{N}_{\underline{P}}(\tau') + E^{N}_{\Omega\underline{P}}(\tau') + E^{N}_{\underline\alpha}(\tau')\right],
\end{split}
\end{align}
where we have applied Theorem \ref{FIdecay} to $\underline{P}$.

Additionally, given $u$ satisfying $\tau'-(r_3)_* \leq u \leq \tau - (r_3)_*$, integration of \eqref{differentialOne} over the region $\overline{C}_{u,\tau'}\cap{\RN{1}}$ implies
\begin{align}\label{transitionOne}
\begin{split}
&\int_{S^2(u,v(u,r_3))}|\underline\alpha|^2 \\
&\leq C(M,\Lambda)\left[\int_{\overline{C}_{u,\tau'}\cap{\RN{1}}}|\underline{P}|^2 + \int_{S^2(u,\tau')}(1-\mu)|\tilde{\underline\alpha}|^2\right].
\end{split}
\end{align}

\subsubsection{Region $\RN{2}$}
Applying a suitable radial weight and H\"{o}lder's inequality, we obtain the differential inequality
\begin{align}\label{differentialTwo}
\begin{split}
\slashed{\nabla}_{L}\left(r|\underline\alpha|^2\right) &= -(1-\mu)|\underline\alpha|^2 + r(1-\mu)\underline\alpha\cdot\underline{P}\\
&\leq -c(M,\Lambda)(1-\mu)|\underline\alpha|^2 + C(M,\Lambda)(1-\mu)|\underline{P}|^2.
\end{split}
\end{align}

For a choice $(u,v)$ in $\RN{2}$ with $\tau' \leq u \leq \tau' - (r_3)_*$, integration of \eqref{differentialTwo} yields
\begin{align}\label{fundamentalEstimateTwo}
\begin{split}
&\int_{S^2(u,v)} |\underline\alpha|^2 + \int_{\overline{C}_{u,\tau'}\setminus{\overline{C}_{u,v}}} |\underline\alpha|^2\\
&\leq C(M,\Lambda)\left[\int_{S^2(u,\tau')} |\underline\alpha|^2 + \int_{\overline{C}_{u,\tau'}\setminus{\overline{C}_{u,v}}}|\underline{P}|^2\right].
\end{split}
\end{align}
Alternatively, for $\tau' - (r_3)_* \leq u \leq \tau - (r_3)_*$ we find
\begin{align}\label{fundamentalEstimateTwoTwo}
\begin{split}
&\int_{S^2(u,v)} |\underline\alpha|^2 + \int_{\left(\overline{C}_{u,\tau'}\setminus{\overline{C}_{u,v}}\right)\cap{\RN{2}}} |\underline\alpha|^2\\
&\leq C(M,\Lambda)\left[\int_{S^2(u,v(u,r_3))} |\underline\alpha|^2 + \int_{\left(\overline{C}_{u,\tau'}\setminus{\overline{C}_{u,v}}\right)\cap{\RN{2}}}|\underline{P}|^2\right]\\
&\leq C(M,\Lambda)\left[\int_{S^2(u,\tau')} (1-\mu)|\tilde{\underline\alpha}|^2 + \int_{\overline{C}_{u,\tau'}\setminus{\overline{C}_{u,v}}}|\underline{P}|^2\right],
\end{split}
\end{align}
where we have used \eqref{transitionOne}.

With the choices $v = v(u,R_3)$ for $\tau' \leq u \leq \tau$ and $v = \tau$ otherwise, integration of \eqref{fundamentalEstimateTwo} and \eqref{fundamentalEstimateTwoTwo} in $u$ yields the integrated decay estimate
\begin{equation}\label{intMedium}
\int_{\mathcal{R}(\tau',\tau) \cap {\RN{2}}} |\underline\alpha|^2 \leq C(M,\Lambda)\left[E^{N}_{\underline{P}}(\tau') + E^{N}_{\Omega\underline{P}}(\tau') + E^{N}_{\underline\alpha}(\tau')\right].
\end{equation}

Likewise, taking $v = \tau$ and integrating in $u$, we obtain
\begin{equation}\label{energyMedium}
\int_{\underline{C}_{\tau,\tau}\cap{\RN{2}}} |\underline\alpha|^2 \leq C(M,\Lambda)\left[E^{N}_{\underline{P}}(\tau') + E^{N}_{\Omega\underline{P}}(\tau') + E^{N}_{\underline\alpha}(\tau')\right].
\end{equation}

Finally, with the choices $u = \tau$ and $v = v(\tau,R_3)$, we have
\begin{align}\label{lossDerivOne}
\begin{split}
&\int_{\overline{C}_{\tau,\tau}\cap{\RN{2}}} |\underline\alpha|^2 \leq \int_{\overline{C}_{\tau,\tau'}\cap{\RN{2}}} |\underline\alpha|^2 \\
&\leq C(M,\Lambda)\left[\int_{S^2(\tau,\tau')} (1-\mu)|\tilde{\underline\alpha}|^2 + \int_{\overline{C}_{\tau,\tau'}\cap{\left(\RN{1}\cup\RN{2}\right)}}|\underline{P}|^2\right]\\
&\leq C(M,\Lambda)\left[E^{N}_{\underline{P}}(\tau') + E^{N}_{\underline\alpha}(\tau')\right],
\end{split}
\end{align}
where we have used the one-dimensional Sobolev inequality and \eqref{Testimates}.

\subsubsection{Region $\RN{3}$}

Integrating the differential inequality
\begin{align}
\begin{split}
\slashed\nabla_{L}\left((1-\mu)r^2|\underline\alpha|^2\right) &= -r^2(1-\mu)\mu_{r}|\underline\alpha|^2 + 2(1-\mu)^2\underline\alpha\cdot\underline{P}\\
&\leq -c(M,\Lambda)(1-\mu)|\underline\alpha|^2 + C(M,\Lambda)(1-\mu)|\underline{P}|^2
\end{split}
\end{align}
with respect to a pair $(u,v)$ in $\RN{3}$, we find
\begin{align}\label{fundamentalEstimateThree}
\begin{split}
&\int_{S^2(u,v)} |\underline\alpha|^2 + \int_{\left(\underline{C}_{u,\tau'}\setminus{\underline{C}_{u,v}}\right)\cap{\RN{3}}} |\underline\alpha|^2\\
&\leq C(M,\Lambda)\left[\int_{S^2(u,v(u,R_3))} |\underline\alpha|^2 + \int_{\left(\underline{C}_{u,\tau'}\setminus{\underline{C}_{u,v}}\right)\cap{\RN{3}}}|\underline{P}|^2\right]\\
&\leq C(M,\Lambda)\left[\int_{S^2(u,\tau')}(1-\mu)|\tilde{\underline\alpha}|^2 + \int_{\underline{C}_{u,\tau'}\setminus{\underline{C}_{u,v}}}|\underline{P}|^2\right],
\end{split}
\end{align}
where we have applied \eqref{fundamentalEstimateTwo} and \eqref{fundamentalEstimateTwoTwo}.  Indeed, the estimate above has the same form as these, and similar arguments yield the integrated decay estimate
\begin{equation}\label{intLarge}
\int_{\mathcal{R}_{\tau',\tau}\cap{\RN{3}}} |\underline\alpha|^2 \leq C(M,\Lambda)\left[E^{N}_{\underline{P}}(\tau') + E^{N}_{\Omega\underline{P}}(\tau') + E^{N}_{\underline\alpha}(\tau')\right],
\end{equation}
and the energy estimate
\begin{align}\label{lossDerivTwo}
\begin{split}
&\int_{\overline{C}_{\tau,\tau}\cap{\RN{3}}} |\underline\alpha|^2 \leq \int_{\overline{C}_{\tau,\tau'}\cap{\RN{3}}} |\underline\alpha|^2 \\
&\leq C(M,\Lambda)\left[\int_{S^2(\tau,\tau')} (1-\mu)|\tilde{\underline\alpha}|^2 + \int_{\overline{C}_{\tau,\tau'}}|\underline{P}|^2\right]\\
&\leq C(M,\Lambda)\left[E^{N}_{\underline{P}}(\tau') + E^{N}_{\underline\alpha}(\tau')\right],
\end{split}
\end{align}
analogous to \eqref{intMedium} and \eqref{lossDerivOne}.

\subsubsection{Derivative Estimates}

Combining the estimates of the previous three subsections, we have deduced the integrated decay estimate
\begin{align}
\begin{split}
&\int_{\mathcal{R}(\tau',\tau) \cap{\{ r\geq 3M\}}} |\underline\alpha|^2 + \int_{\mathcal{R}(\tau',\tau) \cap{\{ r\leq 3M\}}} |\tilde{\underline\alpha}|^2\\
&\leq C(M,\Lambda)\left[E^{N}_{\underline{P}}(\tau') + E^{N}_{\Omega\underline{P}}(\tau') + E^{N}_{\underline\alpha}(\tau')\right],
\end{split}
\end{align}
and the energy estimate
\begin{align}
\begin{split}
&\int_{\overline{C}_{\tau,\tau}} | \underline\alpha|^2 + \int_{\underline{C}_{\tau,\tau}} |\tilde{\underline\alpha}|^2 \\
&\leq C(M,\Lambda)\left[E^{N}_{\underline{P}}(\tau') + E^{N}_{\Omega\underline{P}}(\tau') + E^{N}_{\underline\alpha}(\tau') \right].
\end{split}
\end{align}
It remains to obtain analogous estimates for higher derivatives of $\underline\alpha$.

Using \eqref{underlineP}, we calculate
\begin{equation}
\slashed{\nabla}_{L} \underline\alpha_{A} = -r^{-1}(1-\mu)\underline\alpha_{A} + r^{-2}(1-\mu)\underline{P}_{A}.
\end{equation}
Written with respect to the regular pair \eqref{cosmoHorizonPair}, the above identity leads to
\begin{equation}
|\slashed{\nabla}_{\bar{e}_3} \underline\alpha|^2 \leq C(M,\Lambda)\left[|\underline\alpha|^2 + |\underline{P}|^2\right].
\end{equation}
Integrating and using the results on $\underline\alpha$ from the previous subsections, we deduce the integrated decay estimate
\begin{align}
\begin{split}
\int_{\mathcal{R}(\tau',\tau) \cap{\{ r\geq 3M\}}} |\slashed{\nabla}_{\bar{e}_3} \underline\alpha|^2 \leq &C(M,\Lambda)\int_{\mathcal{R}(\tau',\tau)\cap{\{r \geq 3M\}}} \left[|\underline\alpha|^2 + |\underline{P}|^2\right]\\
&\leq C(M,\Lambda)\left[E^{N}_{\underline{P}}(\tau') + E^{N}_{\Omega\underline{P}}(\tau') + E^{N}_{\underline\alpha}(\tau')\right],
\end{split}
\end{align}
and the energy estimate
\begin{align}
\begin{split}
\int_{\overline{C}_{\tau,\tau}} |\slashed{\nabla}_{\bar{e}_3} \underline\alpha|^2 \leq &C(M,\Lambda)\int_{\overline{C}_{\tau,\tau}} \left[|\underline\alpha|^2 + |\underline{P}|^2\right]\\
&\leq C(M,\Lambda)\left[E^{N}_{\underline{P}}(\tau') + E^{N}_{\Omega\underline{P}}(\tau') + E^{N}_{\underline\alpha}(\tau')\right].
\end{split}
\end{align}

Next we estimate the $\underline{L}$ derivatives of $\underline\alpha$.  Again, owing to issues of regularity, we use the pair \eqref{eventHorizonPair} near the event horizon.  Estimates build upon those of $\underline\alpha$ from the previous subsection, noting the commutation relation $[\slashed{\nabla}_{L},\slashed{\nabla}_{\underline{L}}] = 0$.

In region $\RN{1}$, the differential inequality
\begin{align}
\begin{split}
&\slashed\nabla_{L}\left(r^2(1-\mu)|\slashed\nabla_{e_{4}}\tilde{\underline\alpha}|^2\right) = \slashed{\nabla}_{L}\left(r^2(1-\mu)^{-1}|\slashed{\nabla}_{\underline{L}}\tilde{\underline\alpha}|^2\right)\\
&= 3r^2\mu_{r}(1-\mu)^{-1}|\slashed\nabla_{\underline{L}}\tilde{\underline\alpha}|^2 + 4r^{-1}\slashed\nabla_{\underline{L}}\tilde{\underline\alpha}\cdot\underline{P} + 2(1-\mu)^{-1}\slashed\nabla_{\underline{L}}\tilde{\underline\alpha}\cdot\slashed{\nabla}_{\underline{L}}\underline{P}\\
&-2\left(r^2\mu_{rr}+r\mu_{r}+(1-\mu)\right)\slashed{\nabla}_{\underline{L}}\tilde{\underline\alpha}\cdot\tilde{\underline\alpha}\\
&\leq -c(M,\Lambda)(1-\mu)^{-1}|\slashed\nabla_{\underline{L}}\tilde{\underline\alpha}|^2\\
& + C(M,\Lambda)\left[(1-\mu)|\tilde{\underline\alpha}|^2 + (1-\mu)|\underline{P}|^2 + (1-\mu)^{-1}|\slashed\nabla_{\underline{L}}\underline{P}|^2\right]\\
&= -c(M,\Lambda)(1-\mu)|\slashed\nabla_{e_{4}}\tilde{\underline\alpha}|^2\\
& + C(M,\Lambda)\left[(1-\mu)|\tilde{\underline\alpha}|^2 + (1-\mu)|\underline{P}|^2 + (1-\mu)|\slashed\nabla_{e_{4}}\underline{P}|^2\right]
\end{split}
\end{align}
integrates to give the analogs of \eqref{intSmall} and \eqref{energySmall}.  Likewise, a differential inequality analogous to \eqref{differentialTwo} gives analogs of \eqref{intMedium} and \eqref{energyMedium} on $\RN{2}$.  Summarizing, we have the integrated decay estimate
\begin{equation}
\int_{\mathcal{R}(\tau',\tau) \cap{\{ r\leq 3M\}}} |\slashed{\nabla}_{e_4} \tilde{\underline\alpha}|^2 \leq C(M,\Lambda)\left[E^{N}_{\underline{P}}(\tau') + E^{N}_{\Omega\underline{P}}(\tau') + E^{N}_{\underline\alpha}(\tau')\right],
\end{equation}
and the energy estimate
\begin{equation}
\int_{\underline{C}_{\tau,\tau}} |\slashed{\nabla}_{e_4} \tilde{\underline\alpha}|^2 \leq C(M,\Lambda)\left[E^{N}_{\underline{P}}(\tau') + E^{N}_{\Omega\underline{P}}(\tau') + E^{N}_{\underline\alpha}(\tau')\right].
\end{equation}

Finally, we estimate the angular gradient of $\underline\alpha$.  Again, the estimates build upon those of $\underline\alpha$, noting the commutation relation $[r\slashed\nabla_{A},\slashed{\nabla}_{L}] = 0$.

For example, in region $\RN{1}$ the differential inequality
\begin{align}
\begin{split}
\slashed{\nabla}_{L}\left(r^4|\slashed{\nabla}\tilde{\underline{\alpha}}|^2\right) &= 2r^4\mu_{r}|\slashed{\nabla}\tilde{\underline{\alpha}}|^2 + 2r^2\slashed{\nabla}\tilde{\underline{\alpha}}\cdot\slashed{\nabla}\underline{P}\\
&\leq -c(M,\Lambda)|\slashed{\nabla}\tilde{\underline{\alpha}}|^2 + C(M,\Lambda)|\slashed{\nabla}\underline{P}|^2
\end{split}
\end{align}
leads to the analog of \eqref{intSmall} and \eqref{energySmall}.  Significantly, the estimates in regions $\RN{2}$ and $\RN{3}$ display a loss of derivative, in applying the one-dimensional Sobolev inequality (see \eqref{lossDerivOne} and \eqref{lossDerivTwo}).  Overall, we find the integrated decay estimate
\begin{align}
\begin{split}
&\int_{\mathcal{R}(\tau',\tau) \cap{\{ r\geq 3M\}}} |\slashed{\nabla}\underline\alpha|^2 + \int_{\mathcal{R}(\tau',\tau) \cap{\{ r\leq 3M\}}} |\slashed{\nabla}\tilde{\underline\alpha}|^2\\
&\leq C(M,\Lambda)\left[E^{N}_{\underline{P}}(\tau') + E^{N}_{\Omega\underline{P}}(\tau') + E^{N}_{\underline\alpha}(\tau') + E^{N}_{\Omega\underline\alpha}(\tau')\right],
\end{split}
\end{align}
and the energy estimate
\begin{align}
\begin{split}
&\int_{\overline{C}_{\tau,\tau}} |\slashed{\nabla} \underline\alpha|^2 + \int_{\underline{C}_{\tau,\tau}} |\slashed{\nabla} \tilde{\underline\alpha}|^2 \\
&\leq C(M,\Lambda)\left[E^{N}_{\underline{P}}(\tau') + E^{N}_{\Omega\underline{P}}(\tau') + E^{N}_{\underline\alpha}(\tau') + E^{N}_{\Omega\underline\alpha}(\tau')\right].
\end{split}
\end{align}

\subsubsection{Proof of Decay}

Combining the results of the previous subsections, we have the integrated decay estimate
\begin{equation}
\int_{\tau'}^{\tau} E^{N}_{\underline\alpha}(s)ds \leq C(M,\Lambda)\left[E^{N}_{\underline{P}}(\tau') + E^{N}_{\Omega\underline{P}}(\tau') + E^{N}_{\underline\alpha}(\tau') + E^{N}_{\Omega\underline\alpha}(\tau')\right],
\end{equation}
and the energy estimate
\begin{equation}
E^{N}_{\underline\alpha}(\tau) \leq C(M,\Lambda)\left[E^{N}_{\underline{P}}(\tau') + E^{N}_{\Omega\underline{P}}(\tau') + E^{N}_{\underline\alpha}(\tau') + E^{N}_{\Omega\underline\alpha}(\tau')\right].
\end{equation}

Degenerate exponential decay, analogous to Theorem \ref{degExpDecay}, proceeds in the following manner.  Taking a spherical harmonic decomposition and letting 
\begin{align*}
f(\tau) &= E^{N}_{\underline{P}_{\ell}}(\tau),\\
g(\tau) &=E^{N}_{\underline{\alpha}_{\ell}}(\tau),\\
h(\tau) &= f(\tau) + g(\tau),
\end{align*}
the above estimates on $g$, along with the estimates on $f$ from Theorems \ref{FIboundedness} and \ref{FIdecay}, imply
\begin{align*}
h(\tau) \leq C_1(M,\Lambda)(1+\ell^2)h(\tau'),\\
\int_{\tau'}^{\tau}h(s)ds \leq C_2(M,\Lambda)(1+\ell^2)h(\tau').
\end{align*}
Exponential decay of $h$, hence of $g$, follows from the same argument in Theorem \ref{degExpDecay}.  Note that the exponential decay parameter now degenerates as $\ell^4$, rather than $\ell^2$.  As a consequence, the non-degenerate super-polynomial decay estimate
\begin{equation}
E^{N}_{\underline\alpha}(\tau) \leq C(M,\Lambda,m)\left(\sum_{(q) \leq 2m} \left(E^{N}_{\Omega^{(q)}\underline\alpha}(0) + E^{N}_{\Omega^{(q)}\underline{P}}(0)\right)\right)(1+\tau)^{-m}
\end{equation}
requires twice as much regularity on the initial data.

The $SO(3)$-invariance of the underlying equations leads to higher order versions of the energy estimate above; together with the Sobolev embedding theorems, these higher order estimates lead to pointwise control of $\underline\alpha$.

Analogous results hold for $\alpha$.  Introducing the normalized quantity
\begin{equation}
\tilde{\alpha}_{A} := (1-\mu)^{-1}\alpha_{A},
\end{equation}
regular at the cosmological horizon, and the $N$-energy for $\alpha$
\begin{equation}
E^{N}_{\alpha}(\tau) := \int_{\underline{C}_{\tau,\tau}} J^{N}_{a}[\alpha]\eta^{a} + \int_{\overline{C}_{\tau,\tau}} J^{N}_{a}[\tilde{\alpha}]\eta^{a},
\end{equation}
we have the energy decay
\begin{equation}
E^{N}_{\alpha}(\tau) \leq C(M,\Lambda,m)\left(\sum_{(q) \leq 2m} \left(E^{N}_{\Omega^{(q)}\alpha}(0) + E^{N}_{\Omega^{(q)}P}(0)\right)\right)(1+\tau)^{-m},
\end{equation}
again allowing for a pointwise estimate on $\alpha$ through commutation and application of Sobolev embedding.

\subsection{Decay of $\rho$ and $\sigma$}

The definitions of $P_{A}$ \eqref{P} and $\underline{P}_{A}$ \eqref{underlineP} and the Maxwell equations \eqref{MaxwellOne} and \eqref{MaxwellTwo} yield the relations
\begin{align}\label{Prhosigma}
\begin{split}
&P_{A} = r^2\left(\slashed{\nabla}_{A}\rho + \slashed{\epsilon}_{AB}\slashed{\nabla}^{B}\sigma\right),\\
&\underline{P}_{A} = r^2\left(-\slashed{\nabla}_{A}\rho + \slashed{\epsilon}_{AB}\slashed{\nabla}^{B}\sigma\right).
\end{split}
\end{align}

We present the decay estimates for $\rho$, those for $\sigma$ being analogous.  Considered on the null hypersurface $\underline{C}_{\tau,\tau}$ \eqref{nullHypersurfaces}, the Poincar\'{e} inequality on spheres of symmetry yields
\[\sup_{\underline{C}_{\tau,\tau}}|(\rho - \bar{\rho})(u,\tau,\theta,\phi)|^2 \leq \int_{S^2(u,\tau)} \left[|\slashed{\nabla}\rho|^2 + |\slashed\nabla^2\rho|^2\right].\]
Application of the one-dimensional Sobolev inequality on the hypersurface $\underline{C}_{\tau,\tau}$, along with the relations \eqref{Prhosigma}, yields the estimate
\begin{align*}
\sup_{\underline{C}_{\tau,\tau}}|\rho - \bar{\rho}|^2 &\leq C(M,\Lambda) \left(E^{N}_{P}(\tau) + E^{N}_{\underline{P}}(\tau)+E^{N}_{\Omega P}(\tau) + E^{N}_{\Omega\underline{P}}(\tau)\right)\\
&\leq C(M,\Lambda,m) \left(\sum_{(q) \leq m+1} \left(E^{N}_{\Omega^{(q)}P}(0) + E^{N}_{\Omega^{(q)}\underline{P}}(0)\right)\right)(1+\tau)^{-m},
\end{align*}
where we have appealed to Theorem \ref{PolyDecay}.

A similar result is available on the null hypersurface $\overline{C}_{\tau,\tau}$ \eqref{nullHypersurfaces}.  Taken together, the two yield the super-polynomial decay estimate
\begin{equation}
\sup_{\Sigma_{\tau}} |\rho - \bar{\rho}| \leq C(M,\Lambda,m) \sqrt{\sum_{(q) \leq m+1} \left(E^{N}_{\Omega^{(q)}P}(0) + E^{N}_{\Omega^{(q)}\underline{P}}(0)\right)}(1+\tau)^{-m/2}.
\end{equation}

Likewise, the normalized null component $\sigma$ satisfies
\begin{equation}
\sup_{\Sigma_{\tau}} |\sigma - \bar{\sigma}| \leq C(M,\Lambda,m) \sqrt{\sum_{(q) \leq m+1} \left(E^{N}_{\Omega^{(q)}P}(0) + E^{N}_{\Omega^{(q)}\underline{P}}(0)\right)}(1+\tau)^{-m/2}.
\end{equation}

\subsection{Summary of Results}

Collecting the decay estimates on the Maxwell components and on the higher order quantities $P$ and $\underline{P}$, we summarize our results in the following theorem:

\begin{theorem}
Suppose $F$ is a solution of the Maxwell equations \eqref{MaxwellEqns} on the Schwarzschild-de Sitter spacetime with mass $M$ and cosmological constant $\Lambda$, satisfying the sub-extremal condition \eqref{subextremal}.  Further, suppose that $F$ is specified by smooth initial data on the hypersurface $\Sigma_0$ \eqref{decayFoliation}.  Then the derived quantities $P$ \eqref{P} and $\underline{P}$ \eqref{underlineP} satisfy the Fackerell-Ipser equation \eqref{FIeqn} and the super-polynomial decay estimates of Theorem \ref{PolyDecay}.

In addition, the Maxwell components \eqref{nullMaxwellTensor} satisfy the super-polynomial decay estimates
\begin{align}
\begin{split}
E^{N}_{\alpha}(\tau) &\leq C(M,\Lambda,m)\left(\sum_{(q) \leq 2m} \left(E^{N}_{\Omega^{(q)}\alpha}(0) + E^{N}_{\Omega^{(q)}P}(0)\right)\right)(1+\tau)^{-m},\\
E^{N}_{\underline\alpha}(\tau) &\leq C(M,\Lambda,m)\left(\sum_{(q) \leq 2m} \left(E^{N}_{\Omega^{(q)}\underline\alpha}(0) + E^{N}_{\Omega^{(q)}\underline{P}}(0)\right)\right)(1+\tau)^{-m},
\end{split}
\end{align}
and
\begin{align}
\begin{split}
&\sup_{\Sigma_{\tau}} |\sigma - \bar{\sigma}|\\
&\leq C(M,\Lambda,m) \sqrt{\sum_{(q) \leq m+1} \left(E^{N}_{\Omega^{(q)}P}(0) + E^{N}_{\Omega^{(q)}\underline{P}}(0)\right)}(1+\tau)^{-m/2},\\
&\sup_{\Sigma_{\tau}} |\rho - \bar{\rho}|\\
&\leq C(M,\Lambda,m) \sqrt{\sum_{(q) \leq m+1} \left(E^{N}_{\Omega^{(q)}P}(0) + E^{N}_{\Omega^{(q)}\underline{P}}(0)\right)}(1+\tau)^{-m/2}.
\end{split}
\end{align}

\end{theorem}

\bibliographystyle{plain}
\bibliography{maxwellSchDs}

\end{document}